\documentclass[amsmath,amssymb,article,aps,pra,showpacs,twocolumn,longbibliography]{revtex4-1}
\usepackage[latin1]{inputenc}
\usepackage[T1]{fontenc}
\usepackage{mathptmx}
\usepackage{graphicx}
\usepackage{ae,aecompl}
\usepackage{amssymb}
\usepackage{bm}
\usepackage{extarrows}
\usepackage{tikz-cd}
\usetikzlibrary{matrix}

\usepackage[ruled,vlined,norelsize]{algorithm2e}

\usepackage{color}

\newcommand{\ket}[1]{\left|#1\right\rangle}
\newcommand{\bra}[1]{\left\langle#1\right|}
\newcommand{\braket}[2]{\left\langle#1|#2\right\rangle}

\usepackage{paralist}

\newtheorem{lemma}{Lemma}

\newtheorem{theorem}{Theorem}

\newtheorem{defi}{Definition}

\newcommand{\tr}{\text{tr}}

\def\Ad{\textrm{Ad}\;}

\newcommand{\dd}{\mathfrak{d}}
\newcommand{\sd}{\mathfrak{sd}}

\newenvironment{proof}{\paragraph{Proof:}}{\hfill$\square$}

\begin{document}
\title{A method to determine which quantum operations can be realized with linear optics with a constructive implementation recipe}

\author{Juan Carlos Garcia-Escartin}
\email{juagar@tel.uva.es}  
\affiliation{Universidad de Valladolid, Dpto. Teor\'ia de la Se\~{n}al e Ing. Telem\'atica, Paseo Bel\'en n$^o$ 15, 47011 Valladolid, Spain}

\author{Vicent Gimeno}
\email{gimenov@mat.uji.es}  
 \affiliation{Universitat Jaume I (UJI),  Campus de Riu Sec, Institut Universitari de Matem\` atiques i Aplicacions de Castell\' o , 12071 Castell\' on
de la Plana, Spain.}

\author{Julio Jos\'e Moyano-Fern\'andez}
  \email{moyano@uji.es}
 \affiliation{Universitat Jaume I (UJI),  Campus de Riu Sec, Institut Universitari de Matem\` atiques i Aplicacions de Castell\' o , 12071 Castell\' on
de la Plana, Spain.}

\date{\today}

\begin{abstract}
The evolution of quantum light through linear optical devices can be described by the scattering matrix $S$ of the system. For linear optical systems with $m$ possible modes, the evolution of $n$ input photons is given by a unitary matrix $U=\varphi_{m,M}(S)$ derived from a known homomorphism, $\varphi_{m,M}$, which depends on the size of the resulting Hilbert space of the possible photon states, $M$. We present a method to decide whether a given unitary evolution $U$ for $n$ photons in $m$ modes can be achieved with linear optics or not and the inverse transformation $\varphi_{m,M}^{-1}$ when the transformation can be implemented. Together with previous results, the method can be used to find a simple optical system which implements any quantum operation within the reach of linear optics. The results come from studying the adjoint map bewtween the Lie algebras corresponding to the Lie groups of the relevant unitary matrices.
\end{abstract}

\maketitle



\maketitle
\section{Introduction}\label{intro}
The action of a linear optical device can be described for classical (coherent) light fields as well as for single photon states using unitary matrices. The unitary transformations induced by linear optical devices for generic $n$-photon states of $m$ modes are difficult to compute due to the indistinguishability of photons. Our paper studies the evolution of states of light through a linear optical interferometer acting on $n$ photons in $m$ different modes. These devices are also called linear optics multiports and conserve the total number of photons. 

More precisely, we provide a prescription to assess whether a particular unitary $U$ acting the space of $n$-photon states distributed on $m$ modes can be realized by a linear optical setup, cf. Theorem \ref{teo2.3}, and, if it can, give a recipe to build a device implementing it, cf. Theorem \ref{teo3.1}. Throughout the text we will provide two worked examples of the method.

We will work in the Hilbert space $\mathcal{H}_{m,n}$ of the quantum states $\ket \psi $ of $n$ photons in $m$ modes, with superpositions of the form
$$
\ket \psi = \sum_{n_1+\cdots + n_m=n} \alpha_{n_1,\ldots , n_m} \ket{n_1 \cdots n_m},
$$
where $\ket{n_1\cdots n_m}$ is a state with $n_k$ photons in the $k$-th mode. This space is isomorphic to $\mathbb{C}^M$, where 
$$M=\dim_{\mathbb{C}} \mathcal{H}_{m,n}= {m+n-1\choose n}$$
since each photon number state is orthogonal to the rest:
$$
\langle n_1'\cdots n_m'\vert n_1\cdots n_m\rangle = \delta_{n_1'n_1}\cdots \delta_{n_m' n_m}.
$$
%
Linear interferometers are described by $m\times m$ unitary matrices $S\in U(m)$ (the unitary group of $m\times m$ unitary matrices). 
The interferometer characterized by $S$ naturally induces a transformation $\varphi_{m,M}(S)$ of arbitrary input states
$$
\ket \psi\to\varphi_{m,M}(S)\ket\psi.
$$
In order to understand the action of the interferometer $S$ on multiple photon states we need to study the underlying transformation $\varphi_{m,M}$ from $U(m)$ to the $U(M)$, the group containing the unitary matrix $U$ which gives the evolution of the $n$-photon state in $\mathcal{H}_{m,n}$. Aaronson and Arkhipov \cite{AA11} give a nice algebraic description of the resulting one-to-one correspondence, which turns out to be a group homomorphism we call the $(m,M)$-\emph{photonic homomorphism} (see also \cite{Cai53,SGL04,Sch04} for alternative descriptions).

This paper gives the inverse transformation $\varphi_{m,M}^{-1}(S)$, providing a way to implement any possible linear interferometer. Apart from being interesting in itself, this result also has applications to linear optics quantum computing \cite{KLM01,KMN07} and in boson sampling, a problem which could prove quantum systems can outperform classical computers \cite{AA11,Aar11}. 

\subsection{Organization of the paper}
In Section \ref{main} we give an informal overview of the results. First, Sections \ref{opticalrealizations} and \ref{groupsalgebras} introduce the basic concepts we need. Then, in Section \ref{summary}, we give a general explanation of our results. Section \ref{adjrep} presents the concept of the adjoint representation, which is fundamental in our method. Section \ref{algbasis} describes our preferred basis when working with unitary algebras and Section \ref{exSpace} introduces the state space we will use for our guided examples in the rest of the paper.

Section \ref{achievement} gives a necessary and sufficient condition for a unitary to be implemented with linear optics and Sections \ref{ex2} and \ref{ex3} show concrete examples of operations which cannot and can be implemented with linear optics, respectively.  

Section \ref{construction} shows the method that finds the scattering matrix of the linear system which gives a given unitary operation $U$ (if it exists). Section \ref{recipe} shows the necessary steps, which are illustrated for an example in Section \ref{ex4}.

Section \ref{proofs} gives the detailed proofs of the Theorems on which our results are based. 

Finally, Section \ref{conclusions} discusses the importance of our results and their limitations. 

\section{Outline of the results and definitions}\label{main}

\subsection{Optical realizations and their implementation}\label{opticalrealizations}

Any linear interferometer with $m$ modes is completely described by an $m\times m$ unitary scattering matrix $S$ which has a limited number of degrees of freedom, $2m^2-1$. This means that, except for trivial cases when $n=1$ or $m=1$, linear multiports can only provide a limited subset of all the possible operations over $n$ photons in $m$ modes, which are described by $M\times M$ unitary matrices $U$ with $2M^2-1$ degrees of freedom \cite{MG17}. 

The scattering matrices $S$ are elements of the unitary group $U(m)$ and any general evolution $U$ on $\mathcal{H}_{m,n}$ is an element of the unitary group $U(M)$. The subgroup of all the operations which can be implemented with linear optics is described by the image subgroup of $\varphi_{m,M}$, $\mathrm{im} \varphi_{m,M} = \{B \in U(M) : \exists \, A \in U(m) ~\mbox{such that}~\varphi_{m,M}(A)=B\}$. We say 
\begin{defi}
a matrix $U \in U(M)$ is an $(m,n)$-optical realization if $U \in \mathrm{im} \varphi_{m,M}$.
\end{defi}

In this paper, we relate these groups to give a decision criterion that checks if any particular operator $U$ is in $\mathrm{im} \varphi_{m,M}$ (it is an $(m,n)$-optical realization) and, if it is possible, compute the inverse $\varphi_{m,M}^{-1}(U)$ and recover the unitary matrix $S$ of the linear interferometer which gives the desired evolution. Once we find $S$, we can use previous known results which tell how to build any desired multiport with a fixed scattering matrix using only beam splitters and phase shifters \cite{RZB94} or only beam splitters \cite{BA14,Saw16}, closing the full circle for the experimental implementation of $U$.   

\subsection{Linear optical evolution from the unitary (group) and Hermitian (algebra) matrices point of view}
\label{groupsalgebras}
The map $\varphi_{m,M}$ is a differentiable group homomorphism \cite{AA11} and it induces an algebra homomorphism, $d \varphi_{m,M}$, as described by the commutative diagram 
\begin{center}
	\begin{tikzcd}
		\mathfrak{u}(m) \arrow{r}{d \varphi_{m,M}}\arrow{d}[swap]{\exp} & \mathfrak{u}(M)\arrow{d}{\exp} \\
		U(m)\arrow{r}[swap]{ \varphi_{m,M}} & U(M) \\
	\end{tikzcd}
\end{center}
which relates the unitary groups $U(m)$ and $U(M)$ containing the scattering matrix $S$ and the $n$-photon evolution operator $U$, respectively, to the algebras $\mathfrak{u}(m)$ and $\mathfrak{u}(M)$ whose elements correspond to antihermitian matrices $iH_S$ and $iH_U$ which give an equivalent description of the evolution through exponentiation of the Hamiltonians $H_S$ and $H_U$ ($S=e^{iH_S}$ and $U=e^{iH_U}$) \cite{LN04,GGM18}. 

Both the homomorphism $ \varphi_{m,M}$ and the differential $d \varphi_{m,M}$ can be described in terms of the the photon creation, $\hat{a}_i^{\dag}$, and annihilation, $\hat{a}_i$, operators for mode $k$ \cite{Lou00}, which act on states with $n_k$ photons in the $k$-th mode following
\begin{equation}
\label{CreatAnni}
\begin{split}
\hat{a}_k^{\dagger}\ket{n_k}_k&=\sqrt{n_k+1}\ket{n_k+1}_k,\\
\hat{a}_k\ket{n_k}_k&=\sqrt{n_k}\ket{n_k-1}_k, \hspace{2ex} n \geq 1,\hspace{1ex} \hat{a}_k\ket{0}_k=0.
\end{split}
\end{equation} 

The homomorphism $\varphi_{m,M}$ can be understood from studying the evolution of the creation operators $\hat{a}_k^{\dagger}$ in the Heisenberg picture under the action of a unitary $U$, $\hat{a}_k^{\dagger} \longrightarrow U \hat{a}_k^{\dagger} U^\dag$. For an $n$-photon input state
\begin{equation}
\ket{n_1 n_2 \ldots n_m}=\prod_{k=1}^{m} \left(\frac{\hat{a}_k^{\dag n_k}}{\sqrt{n_k!}}\right)\ket{00\ldots 0},
\end{equation}
the output state after a linear interferometer described from the elements of $S$ as \cite{Cai53,Sch04,SGL04}: 
\begin{equation}
\label{HeisenbergU}
U\ket{n_1 n_2 \ldots n_m}=\prod_{k=1}^{m} \frac{1}{\sqrt{n_k!}}\left(\sum_{j=1}^{m}S_{jk} \hat{a}_j^{\dag}\right )^{n_k}\ket{00\mathellipsis0}.
\end{equation}
We can also write the elements of $U$ from the permanent of different submatrices of $S$ \cite{Sch04}. 

From the differential $d\varphi_{m,M}$  \cite{GGM18}, we can write the effective Hamiltonian $H_U$ of a linear optical transformation as:
\begin{equation}
\label{HUpq}
\bra{p}iH_U\ket{q}=\bra{p}\sum_{l=1}^{m}\sum_{j=1}^{m} iH_{Sjl} \hat{a}^{\dag}_j  \hat{a}_l \ket{q},
\end{equation}
where $\ket{p}$ and $\ket{q}$ are the photon number states in our Hilbert space. The same results can be reached from alternative points of view \cite{LN04,FX97,AC05,BL91,ALN06}.

\subsection{Summary of the results}
\label{summary}
The main design procedure is based on a simple basis decomposition in the image subalgebra of the Hamiltonian, with a detour due to the complications that appear when finding matrix logarithms.

If we know the desired final Hamiltonian, $H_U$, we can check if it can be implemented with linear optics by looking for a linear combination
\begin{equation}
\label{HUbasis}
iH_U=\sum_{i} X_{i} b_i
\end{equation}
of elements of the basis $\left\{b_i \right\}$ of the image subalgebra $\dd\subseteq \mathfrak{u}(M)$ for $i=1, \ldots, m^2$. The elements $b_i=\varphi_{m,M}(a_i)$ are the image of the elements of a basis $\left\{a_i \right\}$ of $\mathfrak{u}(m)$. $H_U$ can be implemented with linear optics if Eq. (\ref{HUbasis}) has a solution. Since $\varphi_{m,M}$ is a linear transformation, the effective Hamiltonian in $\mathfrak{u}(m)$ is
\begin{equation}
iH_S=\sum_{i} X_{i} a_i
\end{equation}
for the same coefficients $X_i$. Now, $S=e^{iH_S}$ is the unitary matrix of the desired linear interferometer, for which there are known methods for an experimental implementation \cite{RZB94,BA14,Saw16}.

The problem reduces to solving Eq. (\ref{HUbasis}), which can be expressed as a system of $M\times M$ linear equations, one for each matrix element, and $m^2$ indeterminates (the size of the basis). If the system is not consistent, we know $H_U$ cannot be implemented exactly using only linear optics.

Unfortunately, finding whether a given unitary $U \in U(M)$ is an optical realization is more involved. In principle, it seems we could just take $iH_U=\log U$ and proceed as before, but, unlike the exponential map, the matrix logarithm is a multivalued function. Computing the matrix logarithm of a unitary numerically presents some challenges, particularly if it has degenerate eigenvalues, as many interesting operations, such as the Quantum Fourier Transform, do. However, there are reliable methods to obtain a Hamiltonian matrix from a unitary \cite{Lor14}.

The greatest obstacle is choosing the correct branch when we are restricted to a subgroup. For a unitary in the image subgroup, we need to guarantee that the logarithm branch we choose is in the image subalgebra $iH_U\in \dd$. Otherwise, we will not be able to find a decomposition in terms of the basis of the image subalgebra, even if $U$ can really be implemented. 

We need a method to find a Hermitian matrix in the form given by Eq. (\ref{HUpq}). These matrices have a strong structure. The only non-zero elements are in positions which correspond to transitions between states that are, at most, one photon away from each other (for input states $\ket{p}$ and $\ket{q}$ which only differ in the photon number in two positions, one mode giving the photon to the other \cite{GGM18}).  

The main contribution of our method is giving an alternative way of finding a suitable basis decomposition, when there is one, by using the adjoint representation. The method avoids using the usual matrix logarithm calculations. There is no need to compute any eigenvalues and we mostly use simple linear algebra methods (matrix multiplication and Gaussian elimination to solve linear systems). If the operation $U$ cannot be implemented, we prove it and if it can, we give a complete description in terms of a linear interferometers.

\subsection{The adjoint representation}\label{adjrep}
We will find the inverse of the $\varphi_{m,M}$ using the adjoint representation, which gives an alternative way of describing linear interferometers and can help to study the evolution of unitary operators \cite{SK17}. 

Let $\mathrm{Ad}_U:\mathfrak{u}(M)\to \mathfrak{u}(M)$ be the adjoint map defined by $\mathrm{Ad}_U(iH_U)=UiH_U U^{\dag}$ \cite{Hal15}. We can also define $\mathrm{Ad}_S:\mathfrak{u}(m)\to \mathfrak{u}(m)$ for the scattering matrix so that $\mathrm{Ad}_S(iH_S)=SiH_S S^{\dag}$. 

When $U$ is an optical realization, the group of linear interferometers can be equally described by $m\times m$ unitaries $S$, by $M\times M$ unitaries $U$ in the image subgroup or by the Hermitian matrices $H_S$ and $H_U$ with $iH_S$ and $iH_U$ in the associated unitary algebra and image subalgebra. Additionally, if $U$ is in the image subalgebra, the adjoint will also describe the same physical system. Finding the inverse in this representation is easier and this is the path we choose.  

The adjoint map is conceptually similar to computing the evolution in the Heisenberg picture. The terms $\hat{a}^{\dag}_j  \hat{a}_l  $ in the effective Hamiltonian given by Eq. (\ref{HUpq}) evolves under the action of the adjoint as
\begin{equation}
U \hat{a}^{\dag}_j  \hat{a}_l U^\dag=U \hat{a}^{\dag}_j U^\dag U \hat{a}_l U^\dag ,
\end{equation}
which, for the definition we use of the adjoint, is to the product of the evolution of the corresponding creation and destruction operators under $U^{-1}=U^\dag$ in the Heisenberg picture. 

In our derivation, we use the fact that, for linear transformations, we can relate the adjoint representations of $S$ and $U=\varphi_{m,M}(S)$ so that
\begin{equation}
\label{homAdj}
\mathrm{Ad}_S(v) = d\varphi_{m,M}^{-1}(\mathrm{Ad}_U(d \varphi_{m,M}(v))), 
\end{equation}
for any $v \in \mathfrak{u}(m)$.

\subsection{Bases for the $\mathfrak{u}(m)$ algebra and the image subalgebra}\label{algbasis}
Consider the canonical basis $\{\ket{1}=\ket{1,0,\ldots , 0},\ket{2}=\ket{0,1,\ldots, 0}, \ldots , \ket{m}=\ket{0,\ldots, 0, 1}  \}$ of $\mathbb{C}^m$. The matrices
\begin{align}
\label{basisdef}
e_{jk}:=& \frac{i}{2}\big (\ket{j}\bra{k}+\ket{k}\bra{j}\big ) \\
f_{jk}:=& \frac{1}{2} \big (\ket{j}\bra{k}-\ket{k}\bra{j}\big ) \nonumber
\end{align}
give a basis of $\mathfrak{u}(m)$. The real linear combinations of the matrices
\begin{align}
\label{basis}
e_{jk} \quad \mbox{for } \, k\leq j=1,\ldots , m\\
f_{jk} \quad \mbox{for } \, k< j=1,\ldots , m,\nonumber
\end{align}
give any desired antihermitian matrix in the algebra.

Observe that $e_{jk}=e_{kj}$ and $f_{jk}=-f_{kj}$. From Eq. (\ref{HUpq}), we see the basis of $\mathfrak{u}(m)$ transforms into
\begin{align*}
d \varphi_{m,M}(e_{jk})=& \frac{i}{2}\big (\hat{a}^\dag_j \hat{a}_k+\hat{a}^\dag_k \hat{a}_j\big ) \neq 0\\
d \varphi_{m,M}(f_{jk})=& \frac{1}{2} \big (\hat{a}^\dag_j \hat{a}_k-\hat{a}^\dag_k \hat{a}_j\big ) \neq 0,
\end{align*}
and therefore the map $d \varphi_{m,M}$ is injective and gives a basis of the image subalgebra.

\subsection{Example space: 5 photons in 2 modes}
\label{exSpace}
To illustrate our results, we will give a few examples using linear interferometers with $m=2$ modes and $n=5$ input photons, for which we have a Hilbert space of dimension $M= {2+5-1\choose 5}=6$. We choose the basis
\begin{equation}
\label{c6basis}
\{\ket{5,0},\ket{4,1},\ket{3,2},\ket{2,3},\ket{1,4},\ket{0,5}\}
\end{equation}
of $\mathbb{C}^6$. For our reference bases, the $i$th basis element corresponds to a column vector filled with zeros except for a single 1 in the $i$th row.

The (2,6)-photonic homomorphism will be denoted as $\varphi_{2,6}:U(2)\to U(6)$. The basis for $\mathbb{C}^2$ will be
\begin{equation}
\label{c2basis}
\left\{\vert 1,0\rangle,\quad \vert 0,1\rangle \right\}.
\end{equation}

We study how ${\rm Ad}_U$ acts on $\dd=d\varphi_{2,6}(\mathfrak{u}(2))$ using the basis $\{e_{11},e_{12},e_{22},f_{12}\}$ of $\mathfrak{u}(2)$,
\begin{equation} \label{u2basis}
\begin{array}{ll}
e_{11}=i\begin{pmatrix}1&0\\
0&0
\end{pmatrix} & e_{12}=\frac{i}{2}\begin{pmatrix}0&1\\
1&0
\end{pmatrix} \\
e_{22}=i\begin{pmatrix}0&0\\
0&1\end{pmatrix} & f_{12}=\frac{1}{2}\begin{pmatrix}0&1\\
-1&0
\end{pmatrix}.
\end{array}
\end{equation}

The $U$ and $iH_U$ matrices in these examples are given for the state order of the basis in Eq. (\ref{c6basis}), i.e. $U_{11}=\bra{1}U\ket{1}=\bra{50}U\ket{50}$, $\ldots$, $U_{32}=\bra{3}U\ket{2}=\bra{32}U\ket{41}$, $\ldots$, $U_{66}=\bra{6}U\ket{6}=\bra{05}U\ket{05}$. The same applies to $S$ and $iH_S$ and the basis in Eq. (\ref{c2basis}).

For a given $iH_S\in \mathfrak{u}(2)$, the corresponding element in the image subalgebra $\mathfrak{d}$, $iH_U=d\varphi_{2,6}(iH_S)$, is given by \cite{GGM18}:
\begin{equation}
\bra{p}iH_U\ket{q}=\bra{p}\sum_{l=1}^{m}\sum_{j=1}^{m} iH_{Sjl} \hat{a}^{\dag}_j  \hat{a}_l \ket{q},
\end{equation}

The matrices $\{a_i\}$ in the basis of $\mathfrak{u}(2)$, ordered as in Eq. (\ref{u2basis}), give us a basis $\{b_1,b_2,b_3,b_4\}$ of the image subspace $\dd\subseteq \mathfrak{u}(M)$, with $b_i=d\varphi_{2,6}(a_i)$ where

\begin{align}
b_1:=&d\varphi_{2,6}(e_{11})=i\hat{a}^\dag_1\hat{a}_1=i\hat{n}_1 \nonumber \\ 
=&i\begin{pmatrix}
5&0&0&0&0&0\\
0&4&0&0&0&0\\
0&0&3&0&0&0\\
0&0&0&2&0&0\\
0&0&0&0&1&0\\
0&0&0&0&0&0
\end{pmatrix}.
\end{align}

\begin{align}
b_2:=&d\varphi_{2,6}(e_{12})=\frac{i}{2}(\hat{a}^\dag_1\hat{a}_2+\hat{a}^\dag_2\hat{a}_1)  \nonumber \\ 
=&\frac{i}{2} \left(
\begin{array}{cccccc}
 0 & \sqrt{5} & 0 & 0 & 0 & 0 \\
 \sqrt{5} & 0 & 2 \sqrt{2} & 0 & 0 & 0 \\
 0 & 2 \sqrt{2} & 0 & 3 & 0 & 0 \\
 0 & 0 & 3 & 0 & 2 \sqrt{2} & 0 \\
 0 & 0 & 0 & 2 \sqrt{2} & 0 & \sqrt{5} \\
 0 & 0 & 0 & 0 & \sqrt{5} & 0 \\
\end{array}
\right).
\end{align}

\begin{align}
b_3:=&d\varphi_{2,6}(e_{22})=i\hat{a}^\dag_2\hat{a}_2=i\hat{n}_2\nonumber \\=&i\begin{pmatrix}
0&0&0&0&0&0\\
0&1&0&0&0&0\\
0&0&2&0&0&0\\
0&0&0&3&0&0\\
0&0&0&0&4&0\\
0&0&0&0&0&5
\end{pmatrix}.
\end{align}

\begin{align}
b_4:= &d\varphi_{2,6}(f_{12})=\frac{1}{2}\left(\hat{a}^\dag_1\hat{a}_2-\hat{a}^\dag_2\hat{a}_1\right) \nonumber \\ 
=&\frac{1}{2} \left(
\begin{array}{cccccc}
 0 &  \sqrt{5} & 0 & 0 & 0 & 0 \\
 -\sqrt{5} & 0 & 2  \sqrt{2} & 0 & 0 & 0 \\
 0 & -2 \sqrt{2} & 0 & 3  & 0 & 0 \\
 0 & 0 & -3  & 0 & 2  \sqrt{2} & 0 \\
 0 & 0 & 0 & -2  \sqrt{2} & 0 &  \sqrt{5} \\
 0 & 0 & 0 & 0 & - \sqrt{5} & 0 \\
\end{array}
\right).
\end{align}

\section{Existence of a unitary evolution via linear optics}\label{achievement}
The adjoint representation gives us a necessary and sufficient condition for the implementation of any given unitary operator $U$ with linear interferometers. 

\begin{theorem}\label{teo2.3}
	$U \in \mathrm{im} \varphi_{m,M} \Longleftrightarrow \mathrm{Ad}_U\mid_{\mathfrak{d}}$ is an automorphism.
\end{theorem}

Theorem \ref{teo2.3} gives a criterion to decide whether a matrix $U\in U(M)$ is an $(n,m)$-optical realization or not. A quantum operation $U$ can be implemented with linear optics if and only if $\mathrm{Ad}_U\mid_{\mathfrak{d}}$ is an automorphism (for any $v$ in the image subalgebra $\mathfrak{d}$, the adjoint $U v U^\dag$ remains in the subalgebra). We only need to see this is the case for a basis of $\mathfrak{d}$, which can be obtained by transforming the elements of a basis of $\mathfrak{u}(m)$ by the algebra homomorphism $d \varphi_{m,M}$.

We can choose a basis of $\mathfrak{u}(m)$ $\{a_1,\cdots,a_m\}$ following Equations (\ref{basisdef}) and (\ref{basis}) so that any element in the algebra can be expressed as a real linear combination of the $a_i$ matrices. For that basis, the matrices $b_i=d\varphi_{m,M}(a_i)$ form a basis for $\mathfrak{d}$.

$U$ is an optical realization if and only if, for any $v$ that is a real linear combination of the computed $b_i$, we can also write $U v U^\dag$ in the same basis. There must exist real coefficients $X_{ij}$ such that
\begin{equation}
\label{system}
U b_i U^\dag=\sum_{j=1}^{m^2} X_{ij}b_j,\quad i=1,\ldots , m^2.
\end{equation}
In order to check whether a given $U$ can be realized with linear optics we need to satisfy $m^2$ equations with $M\times M$ complex matrices, one for each element of the basis, for a total of $m^2M^2$ independent real equations \footnote{Each of the $M\times M$ matrices has $2M^2$ real parameters, but they are antihermitian leaving only $M^2$ independent real values.} with $m^4$ indeterminates. If the system is consistent, $U$ is an $(n,m)$-optical realization. 

\subsection{Example 1: An impossible operation}
\label{ex2}
In our example state space, see Section \ref{exSpace}, in order to decide whether a matrix $U$ is a $(2,6)$-optical realization or not, by Theorem 1, we have to compute $\Ad_U$ and see if $\Ad_U(v)\in \dd$ for any $v \in \dd$. 

The adjoint is linear and it is enough to verify the property for the vectors in the basis $\{b_1,b_2,b_3,b_4\}$ of $\dd$. This leads to a real linear system with $2^2 \cdot 6^2$  equations
\begin{equation}\label{eq:sistema}
\Ad_U(b_j)=\sum_{k=1}^4 X_{jk} b_k,\quad j=1,\ldots , 4
\end{equation}
in the $2^4$ indeterminates $X_{jk}$ belonging to $\mathbb{R}$. If the system is consistent, then $U$ is a $(2,6)$-photonic realization.
\medskip

First, we are going to use Theorem \ref{teo2.3} to show that not every unitary $6\times 6$-matrix is a $(2,6)$-optical realization. If we take the matrix
\begin{equation}
U=\left(
\begin{array}{cccccc}
0& 0&0  &1 & 0 &0 \\
0& 1&0  &0 &0  &0 \\
0&0 &1  &0 & 0 &0 \\
1&0 &0  &0 &0  &0 \\
0& 0& 0 &0 &1  &0 \\
0& 0&  0& 0 &0  &1 \\
\end{array}
\right),
\end{equation}
the system (\ref{eq:sistema}) is inconsistent. Consider, for instance, Eq. (\ref{eq:sistema}) for $b_2$:
\begin{equation}
Ub_2U^\dag=X_{21}b_1+X_{22}b_2+X_{23}b_3+X_{24}b_4,
\end{equation}
which, in matrix form, is
\begin{widetext}
\tiny
\begin{equation}
\left(\begin{matrix}0 & 0 & \frac{3 i}{2} & 0 & \sqrt{2} i & 0\\0 & 0 & \sqrt{2} i & \frac{\sqrt{5} i}{2} & 0 & 0\\\frac{3 i}{2} & \sqrt{2} i & 0 & 0 & 0 & 0\\0 & \frac{\sqrt{5} i}{2} & 0 & 0 & 0 & 0\\\sqrt{2} i & 0 & 0 & 0 & 0 & \frac{\sqrt{5} i}{2}\\0 & 0 & 0 & 0 & \frac{\sqrt{5} i}{2} & 0\end{matrix}\right)=\left(\begin{matrix}5 i X_{21} & \frac{i X_{22}}{2} \sqrt{5} + \frac{\sqrt{5} X_{24}}{2} & 0 & 0 & 0 & 0\\\frac{i X_{22}}{2} \sqrt{5} - \frac{\sqrt{5} X_{24}}{2} & 4 i X_{21} + i X_{23} & \sqrt{2} i X_{22} + \sqrt{2} X_{24} & 0 & 0 & 0\\0 & \sqrt{2} i X_{22} - \sqrt{2} X_{24} & 3 i X_{21} + 2 i X_{23} & \frac{3 i}{2} X_{22} + \frac{3 X_{24}}{2} & 0 & 0\\0 & 0 & \frac{3 i}{2} X_{22} - \frac{3 X_{24}}{2} & 2 i X_{21} + 3 i X_{23} & \sqrt{2} i X_{22} + \sqrt{2} X_{24} & 0\\0 & 0 & 0 & \sqrt{2} i X_{22} - \sqrt{2} X_{24} & i X_{21} + 4 i X_{23} & \frac{i X_{22}}{2} \sqrt{5} + \frac{\sqrt{5} X_{24}}{2}\\0 & 0 & 0 & 0 & \frac{i X_{22}}{2} \sqrt{5} - \frac{\sqrt{5} X_{24}}{2} & 5 i X_{23}\end{matrix}\right).
\end{equation}
\end{widetext}
\normalsize

The system is clearly inconsistent. In the first row we see two constants, $\frac{3 i}{2} $ and $\sqrt{2} i$, which should be equal to zero, which is impossible. 

Therefore, since $\Ad_U$ is not an automorphism of $\dd$, we have that $U\notin \mathrm{Im}(\varphi_{2,6})$ by Theorem 1 and $U$ is not a $(2,6)$-optical realization.

Another way to show that inconsistency is noticing that, if $v\in\dd$, then $\langle n_1', n'_2 \vert v \vert n_1,n_2 \rangle \neq 0$ implies that the input state $\vert n_1,n_2 \rangle$ is, at most, one photon away from the output state $\vert n'_1,n'_2 \rangle$ (cf. Eq. (\ref{HUpq}) and \cite{GGM18}).

This is not the case for the given $U$ and our basis order in Eq. (\ref{c6basis}): notice that $\vert 2,3 \rangle$ is two photons away from $\vert 4,1 \rangle $, but 
\begin{align}
\langle 4,1 \vert \Ad_U(b_2) \vert 2,3 \rangle =& \langle 4,1 \vert U b_2 U^\dag \vert 2,3 \rangle \nonumber \\
=&\langle 4,1 \vert b_2 \vert 5,0 \rangle =\frac{\sqrt{5}}{2}i \neq 0.
\end{align}

\subsection{Example 2: An optical realization}
\label{ex3}
Continuing with \ref{ex2}, let us show that for the operator
\begin{equation}
\label{Upossible}
U=\left(\begin{matrix}\frac{\sqrt{2}}{8} & \frac{\sqrt{10}}{8} & \frac{\sqrt{5}}{4} & \frac{\sqrt{5}}{4} & \frac{\sqrt{10}}{8} & \frac{\sqrt{2}}{8}\\\frac{\sqrt{10}}{8} & \frac{3 \sqrt{2}}{8} & \frac{1}{4} & - \frac{1}{4} & - \frac{3 \sqrt{2}}{8} & - \frac{\sqrt{10}}{8}\\\frac{\sqrt{5}}{4} & \frac{1}{4} & - \frac{\sqrt{2}}{4} & - \frac{\sqrt{2}}{4} & \frac{1}{4} & \frac{\sqrt{5}}{4}\\\frac{\sqrt{5}}{4} & - \frac{1}{4} & - \frac{\sqrt{2}}{4} & \frac{\sqrt{2}}{4} & \frac{1}{4} & - \frac{\sqrt{5}}{4}\\\frac{\sqrt{10}}{8} & - \frac{3 \sqrt{2}}{8} & \frac{1}{4} & \frac{1}{4} & - \frac{3 \sqrt{2}}{8} & \frac{\sqrt{10}}{8}\\\frac{\sqrt{2}}{8} & - \frac{\sqrt{10}}{8} & \frac{\sqrt{5}}{4} & - \frac{\sqrt{5}}{4} & \frac{\sqrt{10}}{8} & - \frac{\sqrt{2}}{8}\end{matrix}\right),
\end{equation}
there exists a matrix $S$ in $U(2)$ such that $\varphi_{2,6}(S)=U$.

We start by solving the system in Eq. (\ref{eq:sistema}). We write one matrix identity for each element in the $\{b_i\}$ basis. For instance, for $b_1$:
\begin{equation}
Ub_1U^\dag=X_{11}b_1+X_{12}b_2+X_{13}b_3+X_{14}b_4,
\end{equation}
becomes
\begin{widetext}
\tiny
\begin{equation}
\left(\begin{matrix}\frac{5 i}{2} & \frac{\sqrt{5} i}{2} & 0 & 0 & 0 & 0\\\frac{\sqrt{5} i}{2} & \frac{5 i}{2} & \sqrt{2} i & 0 & 0 & 0\\0 & \sqrt{2} i & \frac{5 i}{2} & \frac{3 i}{2} & 0 & 0\\0 & 0 & \frac{3 i}{2} & \frac{5 i}{2} & \sqrt{2} i & 0\\0 & 0 & 0 & \sqrt{2} i & \frac{5 i}{2} & \frac{\sqrt{5} i}{2}\\0 & 0 & 0 & 0 & \frac{\sqrt{5} i}{2} & \frac{5 i}{2}\end{matrix}\right)
=\left(\begin{matrix}5 i X_{11} & \frac{i X_{12}}{2} \sqrt{5} + \frac{\sqrt{5} X_{14}}{2} & 0 & 0 & 0 & 0\\\frac{i X_{12}}{2} \sqrt{5} - \frac{\sqrt{5} X_{14}}{2} & 4 i X_{11} + i X_{13} & \sqrt{2} i X_{12} + \sqrt{2} X_{14} & 0 & 0 & 0\\0 & \sqrt{2} i X_{12} - \sqrt{2} X_{14} & 3 i X_{11} + 2 i X_{13} & \frac{3 i}{2} X_{12} + \frac{3 X_{14}}{2} & 0 & 0\\0 & 0 & \frac{3 i}{2} X_{12} - \frac{3 X_{14}}{2} & 2 i X_{11} + 3 i X_{13} & \sqrt{2} i X_{12} + \sqrt{2} X_{14} & 0\\0 & 0 & 0 & \sqrt{2} i X_{12} - \sqrt{2} X_{14} & i X_{11} + 4 i X_{13} & \frac{i X_{12}}{2} \sqrt{5} + \frac{\sqrt{5} X_{14}}{2}\\0 & 0 & 0 & 0 & \frac{i X_{12}}{2} \sqrt{5} - \frac{\sqrt{5} X_{14}}{2} & 5 i X_{13}\end{matrix}\right),
\end{equation}
\normalsize
\end{widetext}
from which we see $X_{11}=\frac{1}{2}$, $X_{12}=1$, $X_{13}=\frac{1}{2}$ and $X_{14}=0$. If we repeat the same operation for each $b_i$ in the basis, we see the system (\ref{eq:sistema}) is consistent and its solutions are given by the matrix
\begin{equation}\label{eq:X}
X=(X_{jk})=\left(
\begin{array}{cccccc}
 1/2 &1 & 1/2 & 0 \\
 1/2& 0 & -1/2 & 0  \\
 1/2 & -1& 1/2 & 0  \\
 0 & 0 & 0 & -1 
\end{array}
\right).
\end{equation}

Therefore, since $\Ad_U$ is a linear map so that $\Ad_U:d\varphi_{2,6}(\mathfrak{u}(2))\to d\varphi_{2,6}(\mathfrak{u}(2))$, we know (by Theorem \ref{teo2.3}) that there exists at least one $S\in U(2)$ such that $
\varphi_{2,6}(S)=U$.

\section{Implementation of the possible operations using linear optics}
\label{construction}
When the operation can be implemented, we can give an explicit implementation which uses the solution to the system of equations. Let $U\in U(M)$, then Theorem \ref{teo2.3} shows that, if $\mathrm{Ad}_U\mid_{\mathfrak{d}}$ is an automorphism, there exists an $S\in U(m)$ such that $\varphi_{m,M}(S)=U$. The goal is to obtain an algorithm which provides this $S$. 

\begin{theorem}\label{teo3.1}
For some $S=\sum_{\ell j} S_{\ell j}\ket{\ell }\bra{j}\in U(m)$, let $\mathrm{Ad}_S:\mathfrak{u}(m)\to \mathfrak{u}(m)$ be the adjoint map, then
there exist $\ell_0,j_0$ such that $-i\bra{\ell_0}\mathrm{Ad}_S(e_{j_0 j_0}) \ket{\ell_0}=\vert S_{\ell_0 j_0}\vert^2 \neq 0$ and 
\begin{equation}
S=e^{i \theta}\sum_{\ell,j}\frac{ \bra{\ell}\mathrm{Ad}_S(f_{jj_0}) \ket{\ell_0}-i\bra{\ell}\mathrm{Ad}_S(e_{jj_0}) \ket{\ell_0}}   
{\sqrt{-i\bra{\ell_0}\mathrm{Ad}_S(e_{j_0 j_0}) \ket{\ell_0}}}\ket{\ell}\bra{j},
\end{equation}
with $\theta \in \mathbb{R}$.
\end{theorem}

All the relevant adjoint operators can be written from $\mathrm{Ad}_S(a_i) $ for the elements $a_i$ of the basis of $\mathfrak{u}(m)$ and the desired $U$. 

Notice that
\begin{align}
d\varphi_{m,M}(\Ad_S(a_i))&=& \nonumber\\
\Ad_U(d\varphi_{m,M}(a_i)))&=&\sum_{j=1}^{m^2} X_{ij}d\varphi_{m,M}(a_i)
\end{align}
for the $X_{ij}$ from the system in Equation (\ref{system}), which must be consistent (otherwise we know the operation cannot be realized). Both $d\varphi_{m,M}$ and $d\varphi_{m,M}^{-1}$ are linear and
\begin{equation}
\label{adjointS}
\Ad_S(a_i)=d\varphi_{m,M}^{-1}(\sum_{j=1}^{m^2} X_{ij}d\varphi_{m,M}(a_i))=\sum_{j=1}^{m^2} X_{ij}a_i
\end{equation}
where all the $X_{ij}$ and $a_i$ are known. 

\subsection{Implementation recipe}\label{recipe}
Given an operator $U$, we first solve the system in Eq. (\ref{system}) (or say it cannot be realized if it is inconsistent). 

Then, we try different integer pairs $\ell, j$ in
\begin{equation}
\label{elements}
\vert S_{\ell j} \vert^2=-i\bra{\ell}\mathrm{Ad}_S(e_{jj}) \ket{\ell}.
\end{equation}
If the chosen $|S_{\ell j}|^2=0$ we have one element of $S$. We continue until we find a pair $\ell_0,j_0$ which gives a nonzero element of $S$. This $S_{\ell_0 j_0}=e^{i\theta}|S_{\ell_0 j_0}|$ will be our reference. 

We can only compute the modulus, but, if we use the same $S_{\ell_0 j_0}$ for all the $\ell,j$ pairs, all the elements of $S$ will have the same global phase, which can be ignored. Using Eq. (\ref{adjointS}) and Theorem \ref{teo3.1}, we can compute all the elements of a scattering matrix $S$ which realizes the desired operator $U$. The scattering matrix, in turn, can be used to build the desired device with linear optical elements \cite{RZB94,BA14,Saw16}. 

\subsection{Implementation example}
\label{ex4}
In Section \ref{ex3} we proved the existence of a matrix $S\in U(2)$ such that $\varphi_{2,6}(S)=U$ for the unitary operation $U$ in Eq.~(\ref{Upossible}). In order to find this matrix $S$ up to global phase we apply Theorem \ref{teo3.1}.

We first look for a nonzero element of $S$ from Eq.~(\ref{elements}):
\begin{equation}
\vert S_{\ell j} \vert^2=-i\bra{\ell}\mathrm{Ad}_S(e_{jj}) \ket{\ell}.
\end{equation}
We start with $\ell=j=1$
\begin{equation}
\vert S_{11} \vert^2=-i\bra{1}\mathrm{Ad}_S(e_{11}) \ket{1}.
\end{equation}
We also need to use Eq.~(\ref{adjointS}) 
\begin{equation}
\Ad_S(a_i)=d\varphi_{m,M}^{-1}(\sum_{j=1}^{m^2} X_{ij}d\varphi_{m,M}(a_i))=\sum_{j=1}^{m^2} X_{ij}a_i,
\end{equation}
which, for $a_1=e_{11}$ and the matrix with the coefficients of the solution in Eq.~(\ref{eq:X}), gives
\begin{align*}
\Ad_S(e_{11})=&X_{11}e_{11}+X_{12}e_{12}+X_{13}e_{22}+X_{14}f_{12}\\ \nonumber
=& \frac{i}{2} \left(\begin{matrix}1 & 1\\ 1 & 1\end{matrix}\right),
\end{align*}
so that
\begin{equation}
\vert S_{11} \vert^2=\frac{1}{2}.
\end{equation}
We obtain $S_{11}=e^{i\theta}\frac{1}{\sqrt{2}}$. From Theorem \ref{teo3.1}, we see the remaining $S_{\ell j}$ are:
\small
\begin{equation}
S_{\ell j}=\sqrt{2}\exp (i \theta) (\bra{\ell}\mathrm{Ad}_S(f_{j1}) \ket{1}-i\bra{\ell}\mathrm{Ad}_S(e_{j1}) \ket{1}).
\end{equation}
\normalsize

We use the same reference $S_{11}$ to find the rest of the entries in $S$:
\small
\begin{align}
S_{12}&=\sqrt{2}\exp (i \theta) (\bra{1}\mathrm{Ad}_S(f_{21}) \ket{1}-i\bra{1}\mathrm{Ad}_S(e_{21}) \ket{1}),\\
S_{21}&=\sqrt{2}\exp (i \theta) (\bra{2}\mathrm{Ad}_S(f_{11}) \ket{1}-i\bra{2}\mathrm{Ad}_S(e_{11}) \ket{1}),\\
S_{22}&=\sqrt{2}\exp (i \theta) (\bra{2}\mathrm{Ad}_S(f_{21}) \ket{1}-i\bra{2}\mathrm{Ad}_S(e_{21}) \ket{1}).
\end{align}
\normalsize
All the elements can be computed from the basis $\{e_{11},e_{12},e_{22},f_{12}\}$ (remembering $f_{ii}=0$, $e_{jk}=e_{kj}$ and $f_{jk}=-f_{kj}$). Apart from $\Ad_S(e_{11})$, we need the matrices
\begin{align}
\mathrm{Ad}_S(f_{21})&=-\mathrm{Ad}_S(f_{12})\\ \nonumber
=&-(X_{41}e_{11}+X_{42}e_{12}+X_{43}e_{22}+X_{44}f_{12})\\ \nonumber
=&f_{12}=\frac{1}{2}\left(\begin{matrix}0 & 1\\ -1 & 0\end{matrix}\right),
\end{align}
\begin{align}
\mathrm{Ad}_S(e_{21})&=\mathrm{Ad}_S(e_{12})\\ \nonumber
=&X_{21}e_{11}+X_{22}e_{12}+X_{23}e_{22}+X_{24}f_{12}\\ \nonumber
=&\frac{1}{2}e_{11}-\frac{1}{2}e_{22}=\frac{i}{2}\left(\begin{matrix}1 & 0\\ 0 & -1\end{matrix}\right),
\end{align}
which give the solution
\begin{align}
S_{12}&=\sqrt{2}\exp (i \theta) \left(0-i \frac{i}{2}\right)=e^{i\theta}\frac{1}{\sqrt{2}},\\
S_{21}&=\sqrt{2}\exp (i \theta) \left(0-i\frac{i}{2}\right)=e^{i\theta}\frac{1}{\sqrt{2}}, \\
S_{22}&=\sqrt{2}\exp (i \theta) \left(-\frac{1}{2}-0\right) \ket{1})=e^{i\theta}\frac{-1}{\sqrt{2}}.
\end{align}

We can ignore the global phase $\theta$ and obtain the scattering matrix
\begin{equation}
S=\frac{1}{\sqrt{2}}\left(\begin{matrix}1& 1\\
1& -1\end{matrix}\right).
\end{equation}
If we compute $\varphi_{2,6}(S)$ using Eq. (\ref{HeisenbergU}), we can check we get the desired evolution $U$.  
 
Notice that, had we tried to simply take the matrix logarithm of the unitary in Eq.~(\ref{Upossible}) we can obtain results such as  
\footnotesize
\begin{equation}
iH_U\!=\!\left(\!\begin{matrix}1.293 i & \!- 0.621 i & \!- 0.878 i & -\!0.878 i & \!- 0.621 i & \!- 0.278 i\\\!- 0.621 i & 0.738 i & \!- 0.393 i & 0.393 i & 0.833 i & 0.621 i\\\!- 0.878 i & \!- 0.393 i & 2.126 i & 0.555 i & \!- 0.393 i & \!- 0.878 i\\\!- 0.878 i & 0.393 i & 0.555 i & 1.015 i & \!- 0.393 i & 0.878 i\\\!- 0.621 i & 0.833 i & \!- 0.393 i & \!- 0.393 i & 2.404 i & \!- 0.621 i\\\!- 0.278 i & 0.621 i & \!- 0.878 i & 0.878 i & \!- 0.621 i & 1.848 i\end{matrix}\!\right),
\end{equation}
\normalsize
which was computed numerically and is presented rounded to three decimal places. While this is a valid logarithm ($U=e^{iH_U}$), the matrix is not in the image subspace (it has nonzero elements for transitions between states which are more than one photon away). As a result, this logarithm is not compatible with our approach using the basis of $\mathfrak{d}$ to check whether $U$ can be implemented or not.

\section{Proofs of the main results}
\label{proofs}
In this Section, we prove Theorems \ref{teo2.3} and \ref{teo3.1}, which lay the foundations for our results.

Let $\mathrm{Ad}_U:\mathfrak{u}(M)\to \mathfrak{u}(M)$ be the adjoint map defined by $\mathrm{Ad}_U(iH_S)=UiH_SU^{\dag}$ \cite{Hal15}.

We denote by $\mathfrak{d}$ the subalgebra $d \varphi_{m,M}(\mathfrak{u}(m))\subseteq \mathfrak{u}(M)$, and  by $\mathfrak{sd}$ the subalgebra $d \varphi_{m,M}(\mathfrak{su}(m))\subseteq \mathfrak{u}(M)$ where $\mathfrak{su}(m)$ is the special unitary algebra of dimension $m$ which gives by exponentiation the matrices in the special unitary group which describe any quantum evolution for a quantum state up to an unobservable global phase shift. Notice that $d \varphi_{m,M}:\mathfrak{u}(m)\to \dd$ is a bijection.

\subsection{Proof of Theorem \ref{teo2.3}}
 
\begin{lemma}\label{tr}
Let $iH_U=d \varphi_{m,M}(iH_S) \in \mathfrak{d}$ for $iH_S \in \mathfrak{u}(m)$, then
$$
\tr (iH_U) = {n+m-1 \choose n-1}\tr(iH_S).
$$
Therefore, $v\in \mathfrak{sd}$ if and only if $\tr(v)=0$. Moreover
$$
\mathfrak{d}=\mathfrak{sd}\oplus \mathrm{span}_{\mathbb{R}}(i I_M).
$$
\end{lemma}

\begin{proof}
$
\tr (iH_U)=\tr (d \varphi_{m,M}(iH_S) ) = \sum_{\ell}^{M} \sum_{j,k}^{m} \langle \ell \vert iH_{Sjk} \hat{a}_j^\dag \hat{a}_k \vert \ell \rangle =
  \sum_{\ell j k} iH_{Sjk}\langle \ell \vert  \hat{a}_j^\dag \hat{a}_k \vert \ell \rangle=  \sum_{\ell j k}iH_{Sjk}\delta_{jk} \langle \ell \vert  \hat{n}_j \vert \ell \rangle  =
\sum_{\ell j } iH_{Sjj}\langle \ell \vert  \hat{n}_j \vert \ell \rangle =\sum_{j } iH_{Sjj}\sum_{\ell} \langle \ell \vert  \hat{n}_j \vert \ell \rangle=\sum_{j}i H_{Sjj}\frac{nM}{m} = {n+m-1 \choose n-1}\tr(iH_S),$
taking into account that the sum for the average of the photon number operator $\hat{n}_j=\hat{a}_j^{\dag}\hat{a}_j$ in mode $j$, $\sum_{\ell} \langle \ell \vert  \hat{n}_j \vert \ell \rangle$, must be the same for every mode (we cover all the permutations for photon occupation). For $n$ total photons and $M$ possible states, the total sum is $nM$, which is divided by $m$ for each position. 

Observe that $\mathrm{span}_{\mathbb{R}}(i I_M)$ is a subalgebra of $\mathfrak{d}$, since
$$
d \varphi_{m,M}(\beta I_m) =\beta\sum_{jk} \delta_{jk}\hat{a}_j^\dag \hat{a}_k = \beta \sum_j \hat{n}_j=\beta n I_M. 
$$

Let $iH_U\in \mathfrak{d}$ and $v=\frac{\tr(iH_U)}{M}I_M\in \mathrm{span}_{\mathbb{R}}(i I_M)$, since $\mathrm{Re}(\tr(iH_U))=0$. Then $iH_U-v\in \sd$ as $\tr(iH_U-v)=0$, hence $iH_U=(iH_U-v)+v\in \mathfrak{sd}+ \mathrm{span}_{\mathbb{R}}(i I_M)$. The result follows from the fact that $\mathfrak{sd}\cap \mathrm{span}_{\mathbb{R}}(i I_M)=0$.

\end{proof}

\begin{lemma}\label{lema:sd}
Let $U\in U(M)$, then $\mathrm{Ad}_U\mid_{\mathfrak{d}}$ is an automorphism if and only if $\mathrm{Ad}_U\mid_{\mathfrak{sd}}$ is an automorphism.
\end{lemma}

\begin{proof}
We first assume that $\mathrm{Ad}_U\mid_{\mathfrak{d}}$ is an automorphism, then for any $v\in \mathfrak{sd}$ 
$$
\tr(\mathrm{Ad}_U(v))=\tr(UvU^\dag)=\tr(v)=0,
$$
therefore $\mathrm{Ad}_U(v)\in \mathfrak{sd}$ by Lemma \ref{tr}. This proves that $\mathrm{Ad}_U$ is an endomorphism. Moreover, the kernel is trivial, since $\Vert \mathrm{Ad}_U(v) \Vert=0$ implies $\Vert v\Vert=0$ for any $v$, as can be checked for the trace norm:
$
\Vert \mathrm{Ad}_U(v) \Vert =  \sqrt{\mathrm{tr}(\mathrm{Ad}_U(v)\cdot (\mathrm{Ad}_U(v))^{\dag})} =  
\sqrt{\mathrm{tr} (UvU^{\dag} (UvU^{\dag})^{\dag})   } = \sqrt{\mathrm{tr} (UvU^{\dag} Uv^{\dag}U^{\dag}) } = \sqrt{\mathrm{tr} (vv^{\dag}) } =\Vert v\Vert .
$

Conversely, let us assume that $\mathrm{Ad}_U\mid_{\mathfrak{sd}}$ is an automorphism. By Lemma \ref{tr}, since $\mathfrak{d}=\mathfrak{sd}\oplus \mathrm{span}_{\mathbb{R}}(i I_M)$ and $\mathrm{Ad}_U:\mathrm{span}_{\mathbb{R}}(i I_M) \to \mathrm{span}_{\mathbb{R}}(i I_M)$ is the identity, then $\mathrm{Ad}_U\mid_{\mathfrak{d}}$ is an automorphism.
\end{proof}

We can now prove Theorem \ref{teo2.3}.

\begin{proof}
$\Longrightarrow$) Let $U \in \mathrm{im} \varphi_{m,M}$, then there exists $iH_U\in \mathfrak{d}$ such that $U=\exp(iH_U)$, hence for any $v\in \dd$
$
\mathrm{Ad}_U(iH_U)=\exp(iH_U) v \exp(-iH_U) = v + [iH_U,v]+\frac{1}{2}[iH_U,[iH_U,v]]+\frac{1}{3!}[iH_U,[iH_U,[iH_U,v]]]+\cdots
$
and clearly $\mathrm{Ad}_U(v)\in \dd$, since $iH_U,v\in \dd$ and the Lie bracket is closed in $\dd$. This proves that $\mathrm{Ad}_U$ is an endomorphism. Moreover, the kernel is trivial, since $\Vert \mathrm{Ad}_U(v) \Vert=0$  implies $\Vert v\Vert=0$ for some $v$ (see proof of Lemma \ref{lema:sd}).

$\Longleftarrow$) Let $U \in U(M)$ such that $\mathrm{Ad}_U\mid_{\mathfrak{d}}$ is an automorphism. There exists $\theta \in \mathbb{R}$ such that $\exp(i\theta)U\in SU(M)$.  Moreover, for any $W\in \varphi_{m,M}(SU(m))$, there exists $w\in \sd$ with $W=\exp(w)$, and
$
\exp(i\theta)U W (\exp(i\theta)U)^\dag = U \exp(w)U^\dag =\exp(UwU^\dag)=\exp(\mathrm{Ad}_U(w)),
$
by Lemma \ref{lema:sd}, $\mathrm{Ad}_U(w)\in \sd$, and so $\exp(\mathrm{Ad}_U(w))\in \varphi_{m,M}(SU(m))$. Finally, by Lemma 9 of \cite{OZ17}  
$$
\exp(i\theta)U=\exp (i\beta)U'
$$
for some $\beta \in \mathbb{R}$ and $U'\in  \varphi_{m,M}(SU(m))$. Hence
$$
U=\exp (i(\beta - \theta )) U' =\exp (i(\beta - \theta ) I_M)\exp (w'),
$$
with $w'\in \sd$. Since $[I_M,w']=0$, then
$$
U=\exp (i(\beta - \theta ) I_M+w').
$$
Since $i(\beta - \theta ) I_M+w'\in \mathfrak{sd}\oplus \mathrm{span}_{\mathbb{R}}(i I_M) $, by Lemma \ref{tr} $i(\beta - \theta ) I_M+w'\in \dd$, therefore $\exp (i(\beta - \theta ) I_M+w')\in  \mathrm{im} \varphi_{m,M}$.
\end{proof}

 \subsection{Proof of Theorem \ref{teo3.1}}
 
 \begin{proof}
Let $S=\sum_{jk} S_{jk}\ket{j}\bra{k}$. The adjoint map acting on the tangent vectors $e_{jk}, f_{jk}$ of $\mathfrak{u}(m)$ gives:

$
\bra{\ell}\mathrm{Ad}_S(e_{jk}) \ket{h}= \bra{\ell}S e_{jk}S^\dag \ket{h}
$
\begin{align*}
=&\sum_{s,t,\mu,\nu} \bra{\ell} S_{st}\ket{s}\bra{t} e_{jk}S_{\mu\nu}^*\ket{\nu}\braket{\mu}{h}\\
=& \frac{i}{2}\sum_{s,t,\mu,\nu}S_{st} S_{\mu\nu}^* \braket{\ell}{s}\bra{t} \big (\ket{j}\bra{k}+\ket{k}\bra{j}\big )  \ket{\nu}\braket{\mu}{h}\\
=& \frac{i}{2}\sum_{s,t,\mu,\nu}S_{st} S_{\mu\nu}^*  \big ( \delta_{\ell s}\delta_{tj}\delta_{k\nu}\delta_{\mu h}+\delta_{\ell s}\delta_{tk}\delta_{j\nu}\delta_{\mu h}  \big )\\
=&\frac{i}{2} \big ( S_{\ell j} S_{ h k}^*  + S_{\ell k} S_{h j}^*  \big ),
\end{align*}
and similarly
$$
\bra{\ell}\mathrm{Ad}_S(f_{jk}) \ket{h}=\frac{1}{2} \big ( S_{\ell j} S_{ h k}^*  - S_{\ell k} S_{h j}^*  \big ).
$$
This allows us to obtain
\begin{equation}\label{eq:ss}
 S_{\ell j} S_{ h k}^* = \bra{\ell}\mathrm{Ad}_S(f_{jk}) \ket{h}-i\bra{\ell}\mathrm{Ad}_S(e_{jk}) \ket{h}
\end{equation}
for all $\ell, j , h , k$ and, for $\ell=h$ and $j=k$
\begin{equation}\label{eq:s2}
\vert S_{\ell j} \vert^2=-i\bra{\ell}\mathrm{Ad}_S(e_{jj}) \ket{\ell}.
\end{equation}
Since the matrix $S$ is unitary, there exists $S_{\ell_0 j_0}\neq 0$ and there is $\theta \in \mathbb{R}$ with $S_{\ell_0 j_0}=\vert S_{\ell_0 j_0} \vert e^{i\theta}$. By \eqref{eq:ss} and \eqref{eq:s2}, 
\begin{equation}
S_{\ell j}=e^{i\theta}\frac{ \bra{\ell}\mathrm{Ad}_S(f_{jj_0}) \ket{\ell_0}-i\bra{\ell}\mathrm{Ad}_S(e_{jj_0}) \ket{\ell_0}}   
{\sqrt{-i\bra{\ell_0}\mathrm{Ad}_S(e_{j_0 j_0}) \ket{\ell_0}}}.
\end{equation}
\end{proof}

\section{Conclusions}\label{conclusions}
In this paper, we give a way to check whether any given linear operator $U$ on $n$ photons in $m$ modes can be implemented with linear optics or not and, if it can, provide a explicit method to find the multiport $S$ which gives the desired operator.

The method tries to write the Hamiltonian corresponding to the desired operator in terms of a linear combination of a basis of the subalgebra of all the possible Hamiltonians.  

In principle, the same analysis with a decomposition in the basis of the image subalgebra could be directly applied to the $H_U$ coming from computing the logarithm of the desired operator matrix $U$. However, computing a suitable matrix logarithm is far from trivial. By using the adjoint representation we guarantee a simple and flexible method for any operator $U$. The computation only involves matrix multiplications and solving a linear system and avoids computing eigenvalues.

This method solves the problem completely for any given $U$. 

There are some limitations to this result worth mentioning. First, it applies only to systems which can be exactly implemented, which, as $n$ and $m$ grow, become a smaller subset of the possible matrices $U$. In many cases we are more concerned with finding the best approximation. In a future work we will present a different method which finds the linear optics evolution which is locally closest to the desired operator in terms of some operator distance. 

Second, that a particular operation cannot be implemented has limited implications to the related quantum information problem of whether a given quantum gate can be implemented with linear optics or not. Apart from encoding issues (a gate might be realized using only a subspace of the possible states), notice that, in linear interferometers, permutations are not trivial. For instance, the Quantum Fourier Transform matrix might be realizable for some mapping of the logical states to the photon states but not for others. 

Taking into account these precautions, the inverse method we have given can be used in quantum optics and quantum information to search for particular quantum tasks or primitives which can be implemented with linear optics, such as particular instances of quantum cloning machines \cite{SIG05} or simple quantum algorithms showing quantum advantage. In general, the framework provided from group theory helps us to understand better the connections between classical and quantum evolution in linear optics.

The first author has been funded by the Spanish Ministerio de Econom\'ia y Competitividad, Project TEC2015-69665-R, MINECO/FEDER, UE. The second author has been partially supported by the Spanish
Government Ministerio de Econom\'ia y Competitividad (MINECO-FEDER), grant MTM2017-84851-C2-2-P, and by Universitat Jaume I, grant UJI-B2018-35. The third author was partially supported by the Spanish Government Ministerio de Econom\'ia, Industria y Competitividad (MINECO), grant PGC2018-096446-B-C22, as well as by Universitat Jaume I, grants P1-1B2015-02 and UJI-B2018-10.
 

\newcommand{\noopsort}[1]{} \newcommand{\printfirst}[2]{#1}
  \newcommand{\singleletter}[1]{#1} \newcommand{\switchargs}[2]{#2#1}

\end{document}